\documentclass[12pt]{article}
\usepackage[utf8]{inputenc}
\usepackage[utf8]{inputenc}
\usepackage[margin=1in]{geometry}
\usepackage{amsmath}
\usepackage{amssymb}
\usepackage{amsthm}
\usepackage{enumerate}
\usepackage[ruled, vlined]{algorithm2e}
\usepackage{verbatim}
\usepackage{qtree}
\usepackage{graphicx}
\usepackage{subcaption}
\usepackage{thmtools}
\usepackage{thm-restate}
\usepackage[style=alphabetic,
citestyle=alphabetic]{biblatex}
\makeatletter
\def\blfootnote{\xdef\@thefnmark{}\@footnotetext}
\makeatother

\newcommand{\A}{\mathcal{A}}
\newcommand{\B}{\mathcal{B}}

\newcommand{\D}{\mathcal{D}}

\newtheorem{definition}{Definition}

\title{Windowed Prophet Inequalities}
\date{\today}
\author{William Marshall\footnote{wfm@stanford.edu} , Nolan Miranda\footnote{mirandan@stanford.edu} , Albert Zuo\footnote{azuo@stanford.edu}}

\begin{document}
\maketitle
\blfootnote{Research supported in part by NSF grant CCF-1763311}
\begin{abstract}
    The prophet inequalities problem has received significant study over the past decades and has several applications such as to online auctions. In this paper, we study two variants of the i.i.d. prophet inequalities problem, namely the windowed prophet inequalities problem and the batched prophet inequalities problem. For the windowed prophet inequalities problem, we show that for window size $o(n)$, the optimal competitive ratio is $\alpha \approx 0.745$, the same as in the non-windowed case. In the case where the window size is $n/k$ for some constant $k$, we show that $\alpha_k < WIN_{n/k} \le \alpha_k + o_k(1)$ where $WIN_{n/k}$ is the optimal competitive ratio for the window size $n/k$ prophet inequalities problem and $\alpha_k$ is the optimal competitive ratio for the $k$ sample i.i.d. prophet inequalities problem. Finally, we prove an equivalence between the batched prophet inequalities problem and the i.i.d. prophet inequalities problem.
\end{abstract}
\section{Introduction}
Recently, there has been renewed interest in the prophet inequalities problem for its applications in online auctions. In this problem, a gambler is presented with $n$ distributions $\D_1, \ldots, \D_n$ offline. Then, online, the gambler sees samples $X_1, \ldots, X_n$ drawn from $\D_1, \ldots, \D_n$ one at a time. Upon seeing $X_i$, the gambler must immediately decide whether to accept that sample, receiving reward $X_i$ and ending the game, or to reject it irrevocably and go on to see $X_{i+1}$. The goal of the gambler is to maximize her expected reward, and her performance is measured against a ``prophet'' who is able to see all the samples at once and deterministically choose the largest sample in the sequence. The ratio of the gambler's reward to the prophet's reward is known as the competitive ratio. 

This problem is commonly motivated by applications to online auction theory. A seller (playing the role of the gambler) attempts to sell a single item. Potential buyers approach the seller one at a time and offer a bid for the item, which the seller must immediately either accept or reject. However, this setup is too restrictive for many real world scenarios. In particular, it seems reasonable a buyer might make her bid available for several time steps before she decides to buy the item elsewhere. Accordingly, we define a relaxation of the i.i.d. prophet inequalities problem, which we call the \emph{window size $k$ prophet inequalities problem}. In this variant, at time step $i$ the gambler may choose to accept any of $X_{i-k+1}, \ldots, X_i$. Similarly, we also define a related variant the \emph{batch size $k$ prophet inequalities problem}. In this variant, the gambler receives $n$ i.i.d. samples from $\D$, but receives $k$ samples at a time. The gambler must decide whether to accept any one of the available samples, ending the game, or to reject all $k$ of the available samples and receive the next batch of $k$ samples. In addition to being another natural relaxation of the i.i.d. prophet inequalities problem, the batched formulation is useful in studying the windowed prophet inequalities problem. 

\subsection{Our Contributions}

Our work has two primary contributions. First, we formulate the \emph{batch size k prophet inequalities problem}, derive its optimal competitive ratio in terms of $k$, and identify distributions for which the gambler can not improve on this competitive ratio. Second, we formulate the \emph{window size $k$ prophet inequalities problem}, an i.i.d. variant of the model used by \cite{lookahead}, derive its optimal competitive ratio when the window size is $o(n)$, and bound its competitive ratio for window sizes that are a constant fraction of the input length.

To start with, we show an equivalence between the gambler's performance in the batched setting and the standard i.i.d. setting (with fewer samples).

\begin{definition}
Let $\alpha_i$ be the optimal competitive ratio for the i.i.d prophet inequalities problem on $i$ samples (as derived in \cite{HK82}).
\end{definition}

\begin{restatable}{thm}{batch}
\label{thm:batch}
Let $BATCH_k$ be the optimal competitive ratio for batch size $k$ prophet inequalities problem. Then for any integer $k \ge 1$,
\[BATCH_{n/k} = \alpha_k.\]
\end{restatable}

Hill and Kertz conjectured that $\alpha_k$ are monotonically decreasing, which would imply that for any integer $k$, the gambler can do strictly better in the batch size $n/k$ setting than the standard i.i.d. size $n$ setting \cite{HK82}. However, as discussed in Section \ref{sec:batch}, this is not the case when the batch size is $o(n)$.

Next, we move on to windowed prophet inequalities. We show that allowing small window sizes does not help the gambler.

\begin{restatable}{thm}{windowon}
\label{thm:o(n)}
For any $k = o(n)$, the optimal competitive ratio for the window size $k$ prophet inequalities problem is $\alpha$.
\end{restatable}

Our main technical result states that if the window size is instead a constant fraction of the total number of samples, then the gambler can achieve an improved competitive ratio, but this improvement is limited.

\begin{restatable}{thm}{windowcf}
\label{thm:windowcf}
Let $WIN_k$ denote the optimal competitive ratio for the window size $k$ prophet inequalities problem. Then $\alpha_k < WIN_{n/k} \le \alpha_k + o_k(1)$.
\end{restatable}

It is notable that the problems studied in this paper all involve drawing i.i.d. samples from a single distribution instead of drawing one sample each from $n$ distributions. If we instead allow the samples to be non-i.i.d, then even with batches or windows, it is impossible to achieve a competitive ratio that is better than $\frac{1}{2}$. Consider the case when $X_1 = 1$ deterministically, $X_2 = X_3 = \ldots = X_{n-1} = 0$ deterministically, and $X_n$ takes value $\frac{1}{\epsilon}$ with probability $\epsilon$ and 0 otherwise. Unless the gambler can see the first and last samples at the same time (i.e. the gambler has the same information as the prophet), she receives expected payoff 1, while the prophet receives expected payoff $2-\epsilon$.

\subsection{Related Work}
Krengel and Sucheston formalized and provided the first optimal stopping rule for the standard prophet inequalities problem \cite{KS78}. Following their work, others have provided optimal stopping rules with additional desirable properties. For instance, \cite{Sam84, KW12} proved that threshold-based stopping rules (which are simple to describe and implement) are optimal in the non-i.i.d. setting, and \cite{RWW} provide stopping rules with low sample complexity. 

Our work focuses on the i.i.d. setting, where Hill and Kertz provided a stopping rule that achieved a competitive ratio of $1-1/e$, as well as an upper bound on the competitive ratio, which later Kertz showed was equal to $\alpha \approx 0.745$ \cite{HK82, kertz86}. Hill and Kertz conjectured that their stopping rule was optimal, and improving upon $1-1/e$ remained an open problem for over 30 years until \cite{abolhassani_1-1/e} refuted their conjecture with a stopping rule that achieved a competitive ratio of $0.738$. Finally, \cite{correa_iid} closed the gap with a stopping rule that achieved competitive ratio $\alpha$, implying that for all $i$, $\alpha_i \ge \alpha$, thus showing that Hill and Kertz' upper bound is tight.

Several variants of the prophet inequalities problem have been examined, such as multiple items \cite{multiple_items}, matchings \cite{GW}, matroid constraints \cite{KW12}, altered sample ordering \cite{random_ordering, optimal_ordering}, and information constraints \cite{AKW, RWW}. Of particular relevance to our setting is \cite{uncertain_supply}, in which the authors examine a multi-item online auction model in which items are available for sale for a limited (and a priori unknown) amount of time. This work relaxes the assumption that the gambler can see all $n$ values if desired, which is complementary to our model, which relaxes the assumption that the gambler can only see a single value at a time. Additionally, some closely related variants of windowed prophet inequalities have been studied. In particular, in 1989 Boshuizen studied the non-i.i.d. windowed prophet inequalities problem, and derived an optimal stopping rule for the non-i.i.d. setting \cite{lookahead}. By extension, this stopping rule is also optimal in our i.i.d. windowed prophet inequalities problem, and we build on Boshuizen's results by explicitly analyzing the competitive ratio in the i.i.d. setting.

The secretary problem is another well known problem in stopping theory which is closely related to the prophet inequalities problem. Where as in the prophet inequalities problem the gambler attempts to maximize her expected reward, in the secretary problem she attempts to maximize the probability of selecting the largest sample in a randomly ordered stream. In this setting \cite{secretary_window} examine a sliding window input model similar to ours, and find and analyze the optimal stopping rule as a function of window size. Additionally, \cite{secretary_patience} study an even more general model where samples disappear after a random amount of time, and characterize optimal stopping rules in this setting. However, the differences between the prophet inequalities problem and the secretary problem make it difficult to apply techniques from these papers to our setting.

\section{Preliminaries}
In the problems studied in this paper, the distribution presented to the gambler is chosen by an \emph{offline adversary}. That is, the adversary chooses the input distribution and the number of samples to be drawn from that distribution. Although more powerful adversaries have been studied in some contexts, the prophet inequalities problem is typically studied with respect to an offline adversary.

For random variable $X$, let $E_k[X] = E[\max\{X_1, \ldots, X_k\}]$, i.e. the prophet's expected reward in the i.i.d. prophet inequalities problem. Let $V_k(X)$ be the gambler's expected reward in the i.i.d. prophet inequalities problem if her input is $k$ independent copies of $X$ and she plays optimally. The Prophet Inequalities problem can be solved by backwards induction. In particular, the optimal algorithm accepts $X_i$ if $X_i > V_{n-i}(X)$, implying that $V_n(X) = E[\max\{X, V_{n-1}(X)\}]$ for all $n > 1$ \cite{Brown72}. This same idea can be applied to solve the batched prophet inequalities problem as well. In particular, the optimal batch size $k$ algorithm accepts the largest sample in batch $i$ if this sample is greater than $V_{n/k - i} (\max\{X_1, \ldots, X_k\})$, i.e. if the sample is larger than the expected reward from declining the sample.

Throughout this paper, it will be useful for us to leverage distributions which are known to be hard in the i.i.d. prophet inequalities problem. In their 1982 paper, Hill and Kertz identified extremal distributions such that for any $n > 0$ and $\epsilon > 0$, there exists a distribution such that the competitive ratio for the i.i.d. Prophet Inequalities problem with $n$ samples from this distribution is at most $\alpha_n + \epsilon$. All such distributions are discrete distributions with finite support. We will refer to these distributions as \emph{extremal distributions} or \emph{hard distributions} for the prophet inequalities problem.

\section{Batched Prophet Inequalities}\label{sec:batch}
In this section we characterize the optimal competitive ratio for the batch prophet inequalities problem by relating it to the i.i.d. prophet inequalities problem. This relationship becomes quite intuitive once we relate input distributions for the two problems.

\begin{restatable}{lem}{lemmaBatchReduc}
Let $D$ be a discrete probability distribution with finite support and let $k$ be a positive integer. Then there exists $D'$ such that if $X'_1, \dots, X'_k$ are drawn independently from $D'$, then $\max\{X'_1, \dots, X'_k\}$ has the same distribution as $X$ drawn from $D$.\label{lemma:probs}
\end{restatable}
\begin{proof}
We proceed by induction on the size of the support of $D$. In the base case if $D$ is supported on a single element, then $D' = D$ satisfies the desired properties. In the inductive case, suppose that the size of the support of $D$ is $n$. Let $v_{\max}$ be the largest value in the support of $D$ and define $p_{\max}$ to be $\Pr_{X \leftarrow D}[X = v_{\max}]$. Define $p'_{\max}$ to be $\Pr_{X' \leftarrow D'}[X' = v_{\max}]$, and choose $p_{\max}'$ such that $(1-p_{\max}')^k = (1-p_{\max})$. Then, $\Pr_{X'_1, \dots, X'_k \leftarrow D'}[\max\{X'_1, \ldots, X'_k\} = v_{\max}] = p_{\max}$.

Let $D_0$ be $D$ without $v$ (normalized accordingly). By the inductive hypothesis, there exists $D_0'$ such that the max of $k$ independent draws from $D_0'$ has the same distribution as one draw from $D_0$. We claim that the distribution
\[
D' = \begin{cases}
v_{\max} & \text{with probability }p'_{\max}\\
D_0' & \text{with probability }1-p'_{\max}
\end{cases}
\]
satisfies the desired property. As shown above, the probability assigned by $D$ to $v$ is the same as the probability that the max of $k$ draws of $D'$ is $v_{\max}$. Consider $w < v_{\max}$ in the support of $D$. Define $q$ to be $\Pr_{Y \leftarrow D_0}[Y = w]$. Then, $\Pr_{X \leftarrow D}[X = w] = (1-p'_{\max})q$. By the inductive hypothesis, the probability that the max of $k$ draws from $D_0'$ is $q$, so the probability that max of $k$ draws from $D'$ is also $(1-p'_{\max})q$ as desired.
\end{proof}

Clearly, we can also do this process in reverse. That is, for any distribution $D'$, there exists a distribution $D$ such that the max of $k$ samples drawn from $D'$ are distributed identically to a single sample from $D$. This allows us to easily switch between working with the max of multiple i.i.d. random variables and working with a single random variable.

\subsection{Proof of Theorem \ref{thm:batch}}

With these results in hand, we move on to prove Theorem \ref{thm:batch}, the equivalence between the batch prophet inequalities problem and the i.i.d. prophet inequalities problem. Recall that the batch size $k$ prophet inequalities problem is the setup where the gambler receives $n$ total i.i.d samples from a distribution $\mathcal{D}$ in batches of $k$ at a time. The gambler is able to accept any of the $k$ samples in the current batch or reject all $k$ forever and receive the next batch of $k$ samples. 

\batch*

\begin{proof}
First, note that $BATCH_{n/k} \ge \alpha_k$, since we can run the optimal algorithm for i.i.d. prophet inequalities problem on the largest sample in each batch. Therefore, we turn our attention to showing that $BATCH_{n/k} \le \alpha_k$.

Let $\D$ be a discrete distribution with finite support and $c$ be the optimal competitive ratio for the i.i.d. prophet inequalities problem with $n$ samples from $\D$. Suppose for the sake of contradiction that there exists some algorithm that achieves competitive ratio greater than $c$ on the batch size $n/k$ prophet inequalities problem. By Lemma 5, we can construct a distribution $\D'$ such that the max of $n/k$ samples drawn from $\D'$ is distributed according to $\D$. Let $\A$ be the optimal algorithm for the batched prophet inequalities problem and consider running $\A$ on samples drawn from $\D'$. By assumption, $\A$ gets competitive ratio strictly greater than $c$. Note that by the definition of the optimal algorithm, any sample which is not a maximum in its batch is ignored by $\A$. Since the maxima from each batch are distributed according to $\D$, we can achieve a competitive ratio strictly greater than $c$ in the standard setting by simulating the optimal batch algorithm on the samples from $\D$, treating each sample as a batch maximum. This creates a contradiction, so the optimal competitive ratio for $\D'$ in the batch setting is at most $c$. Let $\D_1, \D_2, \ldots$ be the sequence of distributions given in \cite{HK82} such that the competitive ratios of i.i.d. prophet inequalities problem with $k$ samples from a distribution from the sequence approaches $\alpha_k$. Letting $\D = \D_i$ and letting $i$ approach infinity completes the proof.

\end{proof}

Next, we state two corollaries which will be useful when we turn to window prophet inequalities.

\begin{restatable}{cor}{batcho(n)}
\label{cor:batcho(n)}
If $b = o(n)$, then $BATCH_b = \alpha$.
\end{restatable}

This follows from Theorem \ref{thm:batch} and the fact that $\lim \alpha_k = \alpha$ \cite{kertz86}. The relationship between the batched and standard settings also allows us to explicitly find hard distributions for the batched prophet inequalities problem.

\begin{restatable}{cor}{batchreduc}
For any $n, k > 0$ with $k \mid n$ and any $\epsilon > 0$, applying the process of Lemma \ref{lemma:probs} to an appropriate extremal distribution from \cite{HK82} gives a distribution $\D$ such that the competitive ratio for the batch size $k$ prophet inequalities problem with $n$ samples from $\D$ is at most $\alpha_{n/k} + \epsilon$.
\end{restatable}

\section{Windowed Prophet Inequalities - Size $o(n)$}
In this section, we present a proof of Theorem \ref{thm:o(n)} through a reduction to batched prophet inequalities. This allows us to leverage our results for batched prophet inequalities and illustrates connections between the two variants of the model. The proof serves as a warm-up for Section~\ref{section:constantfrac}, where these connections will play a key role. For a more direct proof of Theorem~\ref{thm:o(n)}, refer to Appendix~\ref{section:windowo(n)padding}.

\windowon*

We first prove a lemma regarding the position of the reward selected by a stopping algorithm.

\begin{restatable}{lem}{lemmaUniform}
Let $\mathcal{A}$ be a stopping algorithm for the prophet inequalities problem such that $\mathcal{A}$ achieves competitive ratio $a$ as its input $n$ gets large. Let $b = o(n)$ and pick $\delta > 0$. Then there exists stopping algorithm $\mathcal{B}$ such that, as $n$ gets large, $\mathcal{B}$ achieves competitive ratio at least $a(1-\delta)$, and additionally the value selected by $\mathcal{B}$ is (a priori) uniformly random modulo $b$.
\end{restatable}\label{lemma:uniform}

\begin{proof}
First, let $m = n - (n \bmod b) - b$. Notably, as $n$ gets large, $m$ also gets large, since $n \bmod b + b$ is $o(n)$. Then, define a stopping algorithm $\mathcal{C}(s)$ that ignores the first $s$ samples of the input, simulates $\mathcal{A}$ on the next $m$ samples, then ignores the remaining $n-s-m$ samples. Since $m$ grows with $n$, we have that $\mathcal{C}$ achieves a competitive ratio of $a$ on the $m$ samples it does not ignore. Now, for the entire list of $n$ samples, the gambler's expected reward does not change. However, the prophet's expected reward does change when considering the remaining $n-m$ samples. In particular, the prophet achieves at most $\frac{n}{m}$ times the original expected reward on just $m$ samples. Since $m = n - o(n)$, as $n$ gets large, $\frac{n}{m}$ goes to 1, so the prophet achieves the same reward in expectation. Thus, for any $\delta$, $\mathcal{C}$ achieves a competitive ratio of $a(1-\delta)$. Now, let $\mathcal{B}$ choose a uniformly random $s \in \{1, 2, \ldots, b\}$, then run $\mathcal{C}(s)$ with this value. Then $\mathcal{B}$ achieves competitive ratio $a(1-\delta)$, but also chooses a uniformly random element modulo $b$ as desired. 
\end{proof}

In essence, the above lemma relies on the fact that when $b = o(n)$, we can choose an offset $s \in \{1,2,\ldots, b\}$ at the beginning to ignore, and when $n$ gets large this offset does not affect the competitive ratio too much. Now that we have this lemma, we can proceed with our reduction to the batched setup.

\windowon*

\begin{proof}
First, we see that we can achieve a competitive ratio of at least $\alpha$ by ignoring all samples except the first in our window (i.e. only considering samples one at a time) and simulating the standard prophet inequalities algorithm. Thus, we turn our attention to showing that $\alpha$ is an upper bound as well.

Suppose for contradiction that there exists some stopping algorithm $\mathcal{A}$ for the window size $k$ setup that achieves competitive ratio $a = \alpha\left(\frac{1}{1-\epsilon}\right)^2 > \alpha$ for some $\epsilon > 0$. Now, consider the batch size $b = \left\lceil\frac{1}{\epsilon}\right\rceil \cdot k$ prophet inequalities setup. Notably, $b = o(n)$ since $k = o(n)$. By Corollary \ref{cor:batcho(n)}, we know that the optimal competitive ratio for the batch size $b$ prophet inequalities problem is $\alpha_{n/b}$, so as $n$ gets large, this competitive ratio goes to $\alpha$ \cite{kertz86}. We define a stopping rule $\mathcal{B}$ for the batch size $b$ prophet inequalities setup that has access to $\mathcal{A}$. As $\mathcal{B}$ receives each batch, have it feed samples into $\mathcal{A}$ one at a time until either $\mathcal{A}$ accepts a sample, after which $\mathcal{B}$ should immediately accept that same sample, or the current batch runs out, in which case $\mathcal{B}$ receives the next batch and continues. $\mathcal{B}$ only does not accept the same element as $\mathcal{A}$ when $\mathcal{B}$ gives $\mathcal{A}$ one of the first $k-1$ elements of a batch and $\mathcal{A}$ decides to accept an element from the previous batch that is still in the current window. As $n$ gets large, by Lemma \ref{lemma:uniform}, we know that without loss of generality the index of the element selected by $\mathcal{A}$ is uniformly distributed modulo $b$ and that $\mathcal{A}$ achieves competitive ratio of at least $a(1-\epsilon)$. Thus, the probability that $\mathcal{A}$ selects an element outside the current batch is at most $\frac{k-1}{b} < \epsilon$ by construction. Thus, with probability at most $\epsilon$, $\mathcal{B}$ gets reward zero, and with probability at least $1-\epsilon$ $\mathcal{B}$ gets the same reward as $\mathcal{A}$, which by construction gives competitive ratio $\alpha\left(\frac{1}{1-\epsilon}\right)^{2} \cdot(1-\epsilon)$. This is strictly greater than $\alpha$, which contradicts Corollary \ref{cor:batcho(n)}. Thus, there can be no such $\mathcal{A}$, and therefore $WIN_k = \alpha$ for $k = o(n)$.

\end{proof}

\section{Windowed Prophet Inequalities - Size $n/k$}\label{section:constantfrac}
In this section, we provide a proof of Theorem~\ref{thm:windowcf}, an asymptotically tight bound on the gambler's performance in windowed prophet inequalities when the window size is a constant fraction of the total number of samples. We show that for any $k$, $WIN_{n/k}$ is strictly larger than $\alpha_k$, but as $k$ grows, $WIN_{n/k}$ approaches $\alpha_k$. In Appendix \ref{section:windowcftight} we extend these results by proving a tighter upper bound and giving numerical approximations of its implications for different values of $k$.

\windowcf*

\subsection{Lower Bound}
We will first prove that for integer $k$, $WIN_{n/k} > \alpha_k$. By Theorem \ref{thm:batch}, it is equivalent to prove $WIN_{n/k} > BATCH_{n/k}$. In order to limit the classes of distributions that we must consider for this argument, we leverage a discretization technique defined in \cite{HK81} combined with a strengthened version of properties of this distribution proved in \cite{HK81} and \cite{HK82}. We begin with this discretization and its properties.

\begin{definition}{\cite{HK81}}
For random variable $Y$ and constants $0 \le a < b < \infty$, let $Y_a^b$ be the random variable such that $Y_a^b = Y$ if $Y \not \in [a, b]$, $=a$ with probability $(b-a)^{-1} \int_{Y \in [a, b]} (b-Y)$, and $=b$ otherwise $\left(\text{with probability }(b-a)^{-1} \int_{Y \in [a, b]} (Y-a)\right)$.
\end{definition}

Intuitively, $Y_a^b$ can be viewed as the random variable with the same expectation as $Y$ that has the largest possible variance in the range $[a, b]$, as formalized in the following lemma.

\begin{restatable}{lem}{lemma2.2extended}
Let $Y$ be a random variable, $0 \le a < b < \infty$, and $X$ be a random variable independent of both $Y$ and $Y_a^b$. Then 
\begin{enumerate}[i.]
\item $E[Y] = E[Y_a^b]$
\item $E[\max\{X, Y\}] \le E[\max\{X, Y_a^b\}]$
\item Let $p_a$ and $p_b$ be the probabilities that $Y_a^b$ equals $a$ and $b$ respectively. For all $\epsilon_1, \epsilon_2, \epsilon_3 > 0$ and all distributions satisfying $p_a, p_b > \epsilon_1$ and $Pr(X \in [a + \epsilon_2, b - \epsilon_2]) > \epsilon_3$, there exists $\delta > 0$ depending only on $\epsilon_1$, $\epsilon_2$, and $\epsilon_3$ such that $E[\max\{X, Y\}] < E[\max\{X, Y_a^b\}] - \delta$.
\end{enumerate}\label{lemma:2.2_extended}
\end{restatable}
\begin{proof}
The first two items are proven in Lemma 2.2 of \cite{HK81}, so it only remains to strengthen this proof slightly to prove \emph{iii}. Let $\psi_x(y) = \max\{x,  y\}$. If $x \in [a + \epsilon_2, b - \epsilon_2]$, then by the convexity of $\psi$ and the fact that $p_a, p_b > 0$,
\[\int_{Y \in [a, b]} \max\{x, y\} < (b-a)^{-1} \left( \max\{x, a\} \int_{Y \in [a, b]} (b-Y) + \max\{x, b\} \int_{Y \in [a, b]} (Y - a)\right) - c
\]
for some constant $c$. Thus if we let $W$ be the event that $X \in [a + \epsilon_2, b - \epsilon_2]$, then
\begin{align*}
    E\big[\max\{X, Y\} \mid W\big] &= \int_{X \in [a + \epsilon_2, b - \epsilon_2]} \left(\int_{Y \not \in [a, b]} \max\{X, Y\} + \int_{Y \in [a, b]} \max\{X, Y\}\right)\\
    &<  \int_{X \in [a + \epsilon_2, b - \epsilon_2]} \Bigg(\int_{Y \not \in [a, b]} \max\{X, Y\}\\
    &\phantom{=} \quad + (b-a)^{-1} \Big( \max\{X, a\} \int_{Y \in [a, b]} (b-Y)\\
    &\phantom{=} \quad \quad+ \max\{X, b\} \int_{Y \in [a, b]} (Y - a)\Big) - c\Bigg)\\
    &\le E[\max\{X, Y_a^b\} \mid W] - c \cdot \epsilon_3
\end{align*}
So,
\begin{align*}
    E[\max\{X, Y\}] &= E[\max\{X, Y\} \mid W] Pr(W) + E[\max\{X, Y\} \mid \overline{W}] Pr(\overline{W})\\
    &< (E[\max\{X, Y_a^b\} \mid W] - c \cdot \epsilon_3) Pr(W) + E[\max\{X, Y_a^b\} \mid \overline{W}] \Pr(\overline{W})\\
    &\le E[\max\{X, Y_a^b\}] - c \cdot \epsilon_3^2.
\end{align*}
\end{proof}
These properties can be used to prove a modified version of Lemma 2.4 from \cite{HK82}.
\begin{restatable}{lem}{lem2.4extended}
Let $\epsilon_1$, $\epsilon_2$, and $\epsilon_3$ be arbitrary positive constants. Let $X$ be a non-negative random variable such that
\begin{align*}
    V_1(X)^{-1}\int_{X \in [0, V_1(X)]} (V_1(X) - X) &> \epsilon_1 \\ 
    V_1(X)^{-1}\int_{X \in [0, V_1(X)]} X &> \epsilon_1
\end{align*}
and $Pr(X \in [\epsilon_2, V_1(X) - \epsilon_2]) > \epsilon_3$. Then for all $k$,
\begin{enumerate}[i]
    \item $V_k(X) = V_k(X_0^{V_1(X)})$
    \item There exists $\delta > 0$ depending only on $\epsilon_1$, $\epsilon_2$, $\epsilon_3$, and $k$ such that $E_k[X] < E_k[X_0^{V_1(X)}] - \delta$.
\end{enumerate}\label{lemma:2.4-extended}
\end{restatable}
\begin{proof}
The first conclusion follows from induction on $k$ with base case $V_1(X) = V_1(X_0^{V_1(X)})$ and inductive case $V_k(X) = E[\max\{X, V_{k-1}(X)\}] = E[\max\{X_0^{V_1(X)}, V_{k-1}(X_0^{V_1(X)})] = V_k(X_0^{V_1(X)})$, both following directly from Lemma \ref{lemma:2.2_extended}.

For the second conclusion, let $X_1, \ldots, X_k$ be independent copies of $X$. Note that \\$Pr(\max\{X_1, \ldots, X_{k-1}\} \in [\epsilon_2, 1 - \epsilon_2]) \ge \epsilon_3^k$, so we can apply conclusions 2 and 3 of Lemma \ref{lemma:2.2_extended} and induction to get
\begin{align*}
E_k[X] &= E[\max\{X_1, \ldots, X_k\}]\\
&= E[\max\{X, \max\{X_1, \ldots, X_{k-1}\}\}]\\
&< E[\max \{X_0^{V_1(X)}, \max\{X_1, \ldots, X_{k-1}\}\}] - \delta\\
&\le E_k[X_0^{V_1}(X)] - \delta.
\end{align*}
\end{proof}

We are now ready to show that $WIN_{n/k}$ is bounded away from $BATCH_{n/k}$.

\begin{restatable}{lem}{lb}
Let $k \ge 0$ be an integer. Then $WIN_{n/k} > \alpha_k$.
\end{restatable}

\begin{proof}
Consider the following algorithm $\A$ for the window size $n/k$ prophet inequalities problem:

\begin{itemize}
    \item Simulate the optimal batch size $n/k$ algorithm on the input until the algorithm accepts some element, call it $X^*$.
    \item Continue receiving more elements until $X^*$ is the last element in the window.
    \item Accept the largest element in the current window.
\end{itemize}

Let $D$ be an arbitrary distribution, scaled to have mean 1, and $D'$ be the distribution such that $k$ samples from $D$ are distributed identically to a single sample from $D'$. It will suffice to show that $\A$ run on $n$ samples from $D$ achieves strictly greater expected payoff than the optimal batch size $n/k$ algorithm run on $n$ samples from $D$, and that this difference is bounded above zero by some constant independent of $D$.

Notice that the optimal algorithm accepts the largest sample in the current batch precisely when it is greater than some threshold (namely the expected payoff for the remaining rounds). Furthermore, these thresholds depend only on the distribution and the time step. For a given distribution $D$, let $q_i$ be the probability that the largest sample from the $i$th batch is larger than the the $i$th threshold, or equivalently the probability that a sample from $D'$ is greater than $i$th threshold. Then the probability that the maximums of two consecutive batches are both bigger than the $i$th threshold is $q_i^2$. Given that this occurs, the probability that both maximums fall in the same window and the probability that the maximum of the second batch is larger than the maximum of the first batch are both $1/2$. Thus if we let $X_1, X_2, \ldots \sim D'$ independent and $F$ is the cumulative density function for $D'$, then the constructed window algorithm outperforms the optimal batch algorithm by at least

\begin{equation}
    \sum_{i=1}^{k-1} \frac{1}{4} q_i^2E[X_2 - X_1 \mid F^{-1}(1-q_i) < X_2 \le X_1] \prod_{j<i} (1-q_j),\label{eqn:performance_diff}
\end{equation}
which we will call $\Delta$. If $\Delta$ is larger than some constant independent of $D$, then we are done, so assume that for all $\epsilon_1 > 0$, $\Delta < \epsilon_1$ for some choice of $D$. We now use Equation \ref{eqn:performance_diff} to restrict the set of distributions that we must consider. Notice that each component of Equation \ref{eqn:performance_diff} can be interpreted either as a term describing the optimal batch algorithm run on the batch size $n/k$ Prophet Inequalities problem with distribution $D$, or as a term describing the optimal algorithm run on the standard i.i.d. Prophet Inequalities problem with $k$ samples from distribution $D'$. For the analysis that follows, it will be more convenient to use the latter interpretation.

Note that $q_1 < q_2 < \ldots < q_k$ and suppose temporarily that $k > 2$ and $\prod_{j < k-1} (1-q_j)$ is arbitrarily small. Then the probability that the algorithm will accept the third to last sample if it sees it is arbitrarily high. Let $X \sim D'$. The optimal algorithm accepts the third to last element precisely when it is greater than $V_2(X)$, so in this case $Pr(X \ge V_2(X))$ is arbitrarily high. Thus $V_2(X)$ is arbitrarily close to the mean of $D'$ (i.e. 1), so $D'$ is approximately a constant distribution, and the optimal algorithm can achieve a competitive ratio better than $\alpha_k$ for $D'$. Therefore, we assume that $\prod_{j < k-1} (1-q_j)$ is bounded away from zero.

Consider the random variable $X_0^1$, and let $p_0$ and $p_1$ be the probabilities placed on 0 and 1 respectively. We may assume that $p_0$ is bounded above zero, since otherwise either $X_0^1$ takes a single value, or its mean is not 1, a contradiction in both cases. Similarly, if $p_1$ is arbitrarily small, then the optimal algorithm will never accept an element unless it is greater than the mean of $D'$ (since we may assume without loss of generality that the optimal algorithm does not accept elements that are arbitrarily small). Thus the expected difference between the prophet's payoff and the gambler's payoff is at most

\[\sum_{i=1}^{k}
q_i E\big[\max\{X_2, X_3, \ldots\} - X \mid F^{-1}(1-q_i) < X_1 \le \max\{X_2, X_3, \ldots\}\big] \prod_{j <i} (1-q_j),\]

which would be arbitrarily small if $\Delta$ were arbitrarily small. Thus we may assume that for some $\epsilon_2$ independent of $D'$, $p_0, p_1 > \epsilon_2$. Let $\epsilon_3 > 0$. By Lemma \ref{lemma:2.4-extended}, if $Pr(X \in [\epsilon_3, 1-\epsilon_3])$ is bounded above zero, then the competitive ratio for $D'$ is strictly greater than $\alpha_k$, so we assume that $Pr(X \in [\epsilon_3, 1-\epsilon_3])$ is arbitrarily small.

Let $\B$ be the optimal batch algorithm modified such that the threshold for the second to last timestep is $1-\epsilon_3$ instead of 1. We now consider the performance of $\A'$, which we define to be $\A$ modified to use $\B$ instead of the batch algorithm subroutine. Note that $\Delta$ is now the difference between the payoff of $\A'$ and the payoff of $\B$. The payoff achieved by $\B$ is at least the payoff of the optimal batch algorithm minus $\epsilon_3 \cdot p_1$, which can be made arbitrarily close to the performance of the optimal batch algorithm by choosing $\epsilon_3$ small. Let $p_1' = Pr(X \in [1 - \epsilon_3, 1])$, and note that $p_1'$ is bounded above zero. Then

\[\Delta \ge (q_{k-1} + p_1') E[X_2 - X_1 \mid 1-\epsilon_3 < X_1 \le X_2] \prod_{j \le k-1} (1-q).\] 

In particular, it will suffice to show that the expectation is bounded above zero. This expectation is at least
\begin{equation}
E\big[X_2 - X_1 \mid X_2 > 1, X_1 \in [1-\epsilon_3, 1]\big] \cdot Pr( X_2 > 1, X_1 \in [1-\epsilon_3, 1])\label{eqn:expectation}
\end{equation}
By independence of $X_1$ and $X_2$, this is
\[\Big(E[X_2 \mid X_2 > 1]\cdot Pr(X_2 > 1) - E\big[X_1  \mid X_1 \in [1-\epsilon_3, 1]\big] \cdot Pr(X_2 > 1)\Big)  Pr(X_1 \in [1-\epsilon_3, 1]).\]
Recall that $E[X_0^1] = E[X] = 1$ and $E[X_0^1 \mid X_0^1 > 1] = E[X \mid X > 1]$. Thus
\[
E[X \mid X > 1] \frac{Pr(X > 1)}{Pr(X>1) + p_0} = 1,
\]
so $E[X \mid X > 1] Pr(X > 1) = Pr(X > 1) + p_0$. We also have $E[X \mid 1-\epsilon \le X \le 1] Pr(X > 1) \le Pr(X > 1)$, so their difference is at least $p_0$. Since $p_1$ is bounded above zero, so is $Pr(X \in [1-\epsilon, 1])$, so Equation \ref{eqn:expectation} is bounded above zero as desired. Therefore, $\A'$ achieves a better competitive ratio than $\alpha_k$, completing the proof of the lower bound on $WIN_{n/k}$.
\end{proof}
\subsection{Upper Bound}
Now, we show the second half of Theorem~\ref{thm:windowcf}, an upper bound on $WIN_{n/k}$.

\begin{restatable}{lem}{ub}\label{lemma:clean_ub}
Let $k \ge 0$ be an integer. Then $WIN_{n/k} \le \alpha_k + o_k(1)$.
\end{restatable}

\begin{proof}
We prove a cleaner but looser bound that shares several techniques with the tighter bound, which we give in Appendix~\ref{section:windowcftight}. Fix $k$ and let $l < k$. Let $\D_1, \D_2, \ldots$ be the sequence of distributions from \cite{HK82} such that the competitive ratio for the i.i.d. Prophet Inequalities problem on $l$ samples from a distribution in this sequence approaches $\alpha_l$. Let $\epsilon$ be such that the competitive ratio for $\D_i$ is at most $\alpha_l + \epsilon$, and consider the batch size $n/l$ Prophet Inequalities problem on $\D_i$. We will argue that it is possible to use an algorithm for the window $w = n/k$ prophet inequalities problem to approximately solve the batch size $m = n/l$ prophet inequalities problem for $\D_i$, which has competitive ratio approaching $\alpha_l$. Since the value for $n$ is adversarially chosen, we may assume that $l$ and $k$ both divide $n$.

Let $\A$ be an arbitrary stopping rule for the windowed setting, and construct stopping rule $\B$ for the batch setting as follows. Given an input, simulate $\A$ until some element $X^*$ is selected. If it is possible for $\B$ to accept $X^*$, do so. Otherwise accept nothing. Clearly, when $\B$ accepts an element, it receives the same payoff as $\A$. The only time that $\B$ is unable to accept the same element as $\A$ is when $\A$ accepts one of the last $w$ elements of a batch.

We assume without loss of generality that the algorithm does not accept any value that is zero. Let $p$ be the probability that a given batch of inputs is all zeros and $q$ be the probability that for some window of $w$ samples, all of these samples are zero. Then $q^{m/w} = p$ so $q = p^{w/m}$. The probability that the last $w$ elements in every single batch are zero is $p^{wl/m}$. Let $Y$ be the event that the last $w$ elements of every batch are zero, $X \sim \D_i$, $X_\A$ be the element chosen by $\A$, and $X_\B$ be the element chosen by $\B$. Then
\begin{align*}
E[X_\A] &= E[X_\A \mid Y] Pr[Y] + E[X_\A \mid \overline{Y}] Pr[\overline{Y}]\\
        &= E[X_\A \mid Y] Pr[Y] + E[X_\A \mid \overline{Y}](1 - p^{wl/m})
\end{align*}
By construction, we know that $\B$ gets reward at least $E[X_\A \mid Y]Pr[Y]$, so after translating this into a competitive ratio, we have
\[WIN_{n/k} \le \alpha_l + \epsilon + \frac{E[X_\A \mid \overline{Y}]}{E_l[X]}(1 - p^{wl/m}).\]
From \cite{HK82}, we know that $p = (\frac{\gamma}{l-1})^{1/l} < (\frac{0.35}{l-1})^{1/l}$, so substituting this in, we get
\begin{equation}
WIN_{n/k} \le \alpha_l + \epsilon + \frac{E[X_\A \mid \overline{Y}]}{E_l[X]}\left(1 - \left(\frac{0.35}{l-1}\right)^{l/k}\right).\label{eqn:cr_inequality}
\end{equation}

Next we bound $E[X_\A \mid \overline{Y}] / E_l[X]$. The gambler never does better than the prophet, so $E[X_\A \mid \overline{Y}] \le E_l[X \mid \overline{Y}]$. Let $Z$ be the event that all elements in the entire input are zero. Since the expectation of the maximum of a set of samples (i.e. the prophet's reward) is order invariant, we know $E_l[X \mid \overline{Y}] = E_l[X \mid \overline{Z}]$. By the definition of $Z$,
\begin{align*}
E_l[X] &= 0 \cdot Pr(Z) + E_l[X \mid \overline{Z}] \cdot Pr(\overline{Z}) \\
&= \frac{E_l[X]}{Pr(\overline{Z})}
\end{align*}
Additionally, \[Pr(\overline{Z}) = 1 - p^l = 1 - \frac{\gamma}{l-1} > 1 - \frac{0.35}{l-1} = \frac{l - 1.35}{l-1}.\]
Substituting this into Equation \ref{eqn:cr_inequality}, letting $i$ approach infinity, and setting $l = \sqrt{k}$ gives
\[WIN_{n/k} \le \alpha_{\sqrt{k}} + \frac{\sqrt{k} - 1}{\sqrt{k}-1.35}\left(1 - \left(\frac{0.35}{\sqrt{k}-1}\right)^{1/\sqrt{k}}\right).\]
Finally, since $\alpha_k - \alpha_{\sqrt{k}} = o_k(1)$, it follows that $WIN_{n/k} \le \alpha_k + o_k(1)$.
\end{proof}

\begin{figure}[htbp]
\includegraphics[]{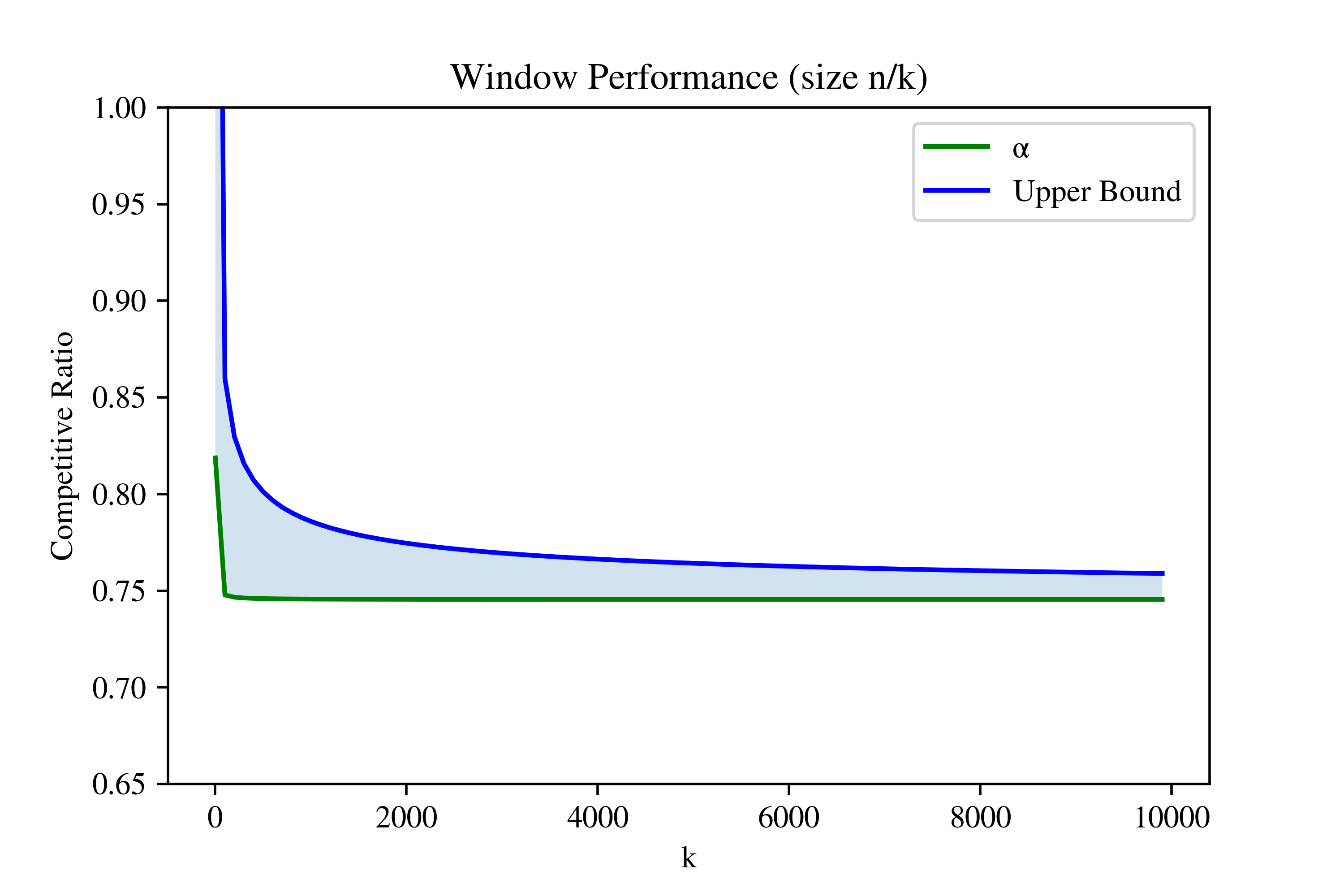}
\caption{The upper and lower bounds for the competitive ratio of the window size $n/k$ prophet inequalities problem.}\label{fig:graph}
\end{figure}

This bound can be tightened by restricting the failure case and optimizing over $l$. The specifics of this bound rely heavily on details about extremal distributions from \cite{HK82} and can be found in Appendix \ref{section:windowcftight}. Figure \ref{fig:graph} shows a comparison of the lower bound and tight upper bound. For example, $0.74785 < WIN_{n/100} \le 0.86095$, $0.74568 < WIN_{n/1000} \le 0.78592$, and $0.74546 < WIN_{n/10,000} \le 0.75885$.

\section*{Acknowledgments}
The authors would like to thank Omer Reingold for his guidance and technical insights, including the idea behind the proof in Appendix~\ref{section:windowo(n)padding}, and Aviad Rubenstein and Matt Weinberg for their helpful comments on earlier drafts of the paper.

\newpage
\printbibliography
\newpage

\begin{appendix}
\section{Alternate Proof of Theorem~\ref{thm:o(n)}}
\label{section:windowo(n)padding}
One way to analyze the $o(n)$ case is to ``pad'' the input distribution such that with high probability no non-zero samples appear in the same window. If a window stopping rule does not see more than one ``interesting'' sample at a time, it can not do better than a stopping rule that only sees a single sample at a time.

\windowon*

\begin{proof}
First, let $p$ denote a probability in $[0,1]$ (to be chosen later to facilitate our proof). Let $m = np$, and let $D_m$ be a discrete distribution with finite support and $c$ be the optimal competitive ratio for the i.i.d. prophet inequalities problem $m$ samples drawn from $D_m$. Now, define $D_m'$ to be the ``zero-padded" version of $D_m$; that is, let $D_m'$ be the distribution that with probability $p$ draws a sample from $D_m$ and takes the value 0 with probability $1-p$. We will analyze the competitive ratio for the window size $k$ prophet inequalities setup with $n$ samples drawn from $D_m'$.

First, suppose the gambler encounters a sample drawn from $D_m$ (which occurs with probability $p$ for a given index). The probability that the gambler encounters another sample drawn from $D_m$ in the next $k-1$ samples (i.e. in the same window starting from the first sample from $D_m$) is at most $kp$ by the union bound. Thus, with probability at most $kp^2$, a given sample is both drawn from $D_m$ and within $k$ upcoming samples of another sample drawn from $D_m$. Therefore, by the union bound, with probability at most $nkp^2$, there is a sample that is both drawn from $D_m$ and within $k$ samples of another sample from $D_m$. 

We want to choose $p$ such that $nkp^2$ is $o(1)$, while $m = np$ is $\omega(1)$. Since $k$ is $o(n)$, we can choose such a $p$.

With probability at least $1 - nkp^2$, no two samples from $D_m$ will be visible within the same window. Using our choice of $p$, as $n$ gets large, this probability goes to 1, and importantly $m$ gets large as well. In the $n$ samples drawn in the window setup described here, the expected number of samples drawn from $D_m$ is $np = m$. Let $m^\ast$ be the number of samples actually drawn from $D_m$. We can use Chernoff bounds to bound the probability of deviation from this expectation. This yields the bounds
$$\operatorname{Pr}\left[m^{*} \geq m\left(1+m^{-1 / 4}\right)\right] \leq e^{-\sqrt{m} / 2}$$
$$\operatorname{Pr}\left[m^{*} \leq m\left(1-m^{-1 / 4}\right)\right] \leq e^{-\sqrt{m} / 3}.$$
Now that we have these, we can analyze the competitive ratio of the gambler given that she draws $m^\ast$ samples from $D_m$ in the windowed prophet inequalities setup on $n$ samples (for large $n$) with window size $k = o(n)$. Since $m$ gets large as $n$ gets large, we have that by our Chernoff bounds the probability that $(1-m^{-1/4})m \le m^* \le (1+m^{-1/4})m$ goes to 1. We now split into two cases. Suppose $m \leq m^{*} \leq m\left(1+m^{-1 / 4}\right)$. Then the gambler achieves a factor of at most $1+m^{-1/4}$ times her original reward. However, $m^{-1/4}$ is $o(1)$ as $n$ gets large, so the gambler's competitive ratio remains the same. Similarly, if $m\left(1-m^{-1 / 4}\right) \leq m^{*} \leq m$, then the prophet could potentially do worse by a factor of $\left(1-m^{-1 / 4}\right)$, but as $n$ gets large, the $m^{-1/4}$ term goes to 0. Thus, in either case, the competitive ratio is changing by a factor of at most 1 - $o(1)$, which ensures the same competitive ratio with $m^\ast$ samples from $D_m$ as $m$ samples from $D_m$. 

We can now conclude that the competitive ratio of the windowed prophet inequalities problem with $n$ samples drawn from $D_m'$ is $c$ (the same as the competitive ratio for the i.i.d. prophet inequalities problem with $m$ samples drawn from $D_m$). If the gambler could achieve some competitive ratio strictly greater than $c$ in this setup, she could receive a competitive ratio strictly greater than $c$ in the i.i.d. prophet inequalities problem by padding her input with zeros and using the window stopping rule to select her element. This is impossible, so the competitive ratio for the windowed prophet inequalities setup is bounded above by $c$. As mentioned before, we can certainly achieve a competitive ratio of at least $c$ by simulating the standard prophet inequality algorithm on the input. Let $\D_{m_1}, \D_{m_2}, \ldots$ be the sequence of distributions given in \cite{HK82} such that the competitive ratios of i.i.d. prophet inequalities problem with $m$ samples from a distribution from the sequence approaches $\alpha_m$. Letting $D_m = \D_{m_i}$ and letting $i$ approach infinity causes $c$ to approach $\alpha$. Thus, $WIN_k$ for window size $k=o(n)$ is $\alpha$, as desired.
\end{proof}
\section{Tighter Upper Bound for Theorem~\ref{thm:windowcf}}
\label{section:windowcftight}
In this section, we make a slight modification to the proof of Lemma \ref{lemma:clean_ub} to make it tighter. Let $l$, $k$, and $\A$ be as defined previously. We consider a similar reduction. Note that in the proof of the cleaner bound, we defined the failure case $\overline Y$ to be the event that at least one of the last $w$ elements in a batch is nonzero. However, the more accurate failure case is when $\A$ (an arbitrary stopping rule for the windowed setting) decides to accept an element, but $\B$ (the constructed batch stopping rule from $\A$) can not. In particular, we can assume that within batch $l-i+1$ ($i$th to last batch), $\A$ will only accept a sample larger than $V_{i-2}(X)$ since we can guarantee this payoff by running the optimal batch algorithm starting at batch $l-i+1$. We can thus change the failure case to be $\overline{Y}'$, the event that for any $i$ at least one of the last $w$ elements in batch $l-i+1$ is larger than $V_{i-2}(X)$. 

As before but replacing $\overline Y$ with $\overline Y'$, we have 
\begin{align*}
WIN_{n/k}   &\le \alpha_l +\epsilon + \frac{E[X_\A \mid \overline{Y}']}{E_l[X]} Pr[\overline{Y}'] \\
        &\le \alpha_l + \epsilon + \frac{E_l[X \mid \overline{Y}']}{E_l[X]} Pr[\overline{Y}'] \\
\end{align*}

Define $Z'_i$ to be the event that one of the samples at the end of batch $l-i+1$ is strictly larger than $v_{i-2}$ (so $\overline{Y}' = \cup_{i=1}^l Z'_i$). Then
\[
E_l[X \mid \overline{Y}'] Pr[\overline{Y}']   = \sum_{j} E_l[X \mid \overline{Z}_{j+1}, \ldots, \overline{Z}_l, Z_j] Pr(Z_j) \prod_{f > j} Pr(\overline{Z_f})
\]
Note that $E_l[X \mid \overline{Z}_{j+1}, \ldots, \overline{Z}_l, Z_j] \le E_l[X \mid {Z}_j]$. Let $D_{l, \delta}$ be the extremal distribution as defined in \cite{HK82} such that the competitive ratio for $l$ samples from $D_{l, \delta}$ is at most $\alpha_l + \frac{\delta \alpha_l^2}{1-\delta \alpha_l}$. Let $D'$ be the distribution such that the max of $n/k$ samples from $D'$ is distributed identically to the max of $n/l$ samples from $D_{l, \delta}$. Let $v_1, \ldots, v_k$ be the values that a sample from $D'$ might take, and $s_i$ be the probability that a sample from $D'$ is less than or equal to the $v_i$. Then

\[
E_l[X \mid {Z}_j] \le \sum_{i=j-2}^l v_i \left(s_i^{k-1} \left(\frac{s_i - s_{j-2}}{1 - s_{j-2}}\right) - s_{i-1}^{k-1} \left(\frac{s_{i-1} - s_{j-2}}{1 - s_{j-2}}\right)\right)
\]
Substituting this into our initial inequality for $WIN_{n/k}$, we have
\[
WIN_{n/k} \le \alpha_l + \frac{\delta\alpha_l^2 }{1-\delta\alpha_l } + \frac{1}{E_l[X]} 
\sum_j \sum_{i=j-2}^l v_i \big(s_i^{k-1} \left(s_i - s_{j-2} \right) - s_{i-1}^{k-1} \left(s_{i-1} - s_{j-2} \right) \big)
\prod_{f > j} s_{f-1}
\]

This holds for any choice of $l$ and $\delta$, and \cite{HK82} describe how to approximate each $s_i$ and $v_i$ as well as $E_l[X]$ and $\alpha_l$. Figure \ref{fig:graph} approximates this expression and optimizes the choice of $l$ to compute an upper bound on the competitive ratio for the windowed prophet inequalities problem. 
\end{appendix}
\end{document}